\documentclass[journal,12pt,draftclsnofoot,onecolumn]{IEEEtran}
\usepackage{amsmath}
\usepackage{graphicx}
\usepackage{epstopdf}
\usepackage{amssymb}
\usepackage{amsthm}
\usepackage{caption}
\usepackage{cite}
\usepackage{subfigure}
\usepackage{enumerate}
\usepackage{textcomp}
\usepackage{tikz}
\usepackage{subfigure}
\usetikzlibrary{patterns}
\usetikzlibrary{shapes,arrows}
\usepackage{verbatim}
\usetikzlibrary{positioning}
\usetikzlibrary{shapes.geometric}
\usetikzlibrary{shapes.misc}
\begin{document}
\title{Impact of Correlation between Nakagami-m Interferers on Coverage Probability and Rate in Cellular Systems}
\author{Suman Kumar \hspace*{1.0in} Sheetal Kalyani  \\
\hspace{0in} Dept. of Electrical Engineering \\
      \hspace{0in}IIT Madras \\
  \hspace{-0.7in}     Chennai 600036, India   \\ 
{\tt \{ee10d040,skalyani\}@ee.iitm.ac.in}\\
}

\maketitle
\begin{abstract} 
Coverage probability and rate expressions are theoretically compared for the following cases: $(i).$ Both the user channel and the $N$ interferers are independent and non identical Nakagami-m distributed random variables (RVs). $(ii).$ The  $N$ interferers are correlated Nakagami-m RVs. It is  analytically shown that  the coverage probability in the presence of correlated interferers is greater than or equal to the coverage probability in the presence of non-identical  independent interferers  when the  shape parameter of the channel between the user and its base station  is not greater than one. It is  further analytically shown that the average rate in the presence of correlated interferers is greater than  or equal to the average rate in the presence of non-identical  independent interferers.  Simulation results are provided  and these match with the obtained theoretical results. The utility of our results are also discussed.
\end{abstract}
\section{Introduction}
Performance degradation of wireless communication is typically caused by multipath fading and co-channel interference. Various fading models have been studied in literature for modeling the interferers and user channels. Among them, the Nakagami-m distribution is a very popular fading model, and  Rayleigh fading can be treated as a special case of Nakagami-m fading \cite{goldsmith2005wireless}.

Coverage probability\footnote{It is a probability that a user can achieve a target Signal-to-Interference-plus-noise-Ratio (SINR) $T$, and outage probability is the complement of  coverage probability.} is an important metric for performance evaluation of cellular systems and for Nakagami-m fading environment it has been studied extensively  for both independent and correlated interferers case \cite{AbuDayya, zhang, 406635, Tellambura, alouini, Trigui, Hadzialic, QiuyanLiu, 905889}. In the case where the fading parameters are arbitrary and possibly non-identical for the $N$ Nakagami-m interferers and the fading parameter for the user channel is also arbitrary, coverage probability expression has been derived in terms of integral in \cite{QiuyanLiu}, infinite series in \cite{Tellambura, alouini, Hadzialic} and multiple series in \cite{905889}. 

Typically, in practical scenario correlation exists among the interferers \cite{correlation, 944855, 991146, 5590312}.   For example in cellular networks when two base stations (BSs) from adjacent sectors act as interferers, the interferers are correlated and it is mandated that while performing system level simulation, this correlation be explicitly introduced in the system \cite{3gpp}. Considering the impact of correlation in the large scale  shadowing component and the small scale  multipath component is also an essential step towards modeling the channel.  The decorrelation distance in multipath components is lower when compared to shadowing components since shadowing is related to terrain configuration and/or large obstacles between transmitter and receiver \cite{991146}.

For more general fading distributions, namely the $\eta-\mu$ fading \cite{4231253} which also includes Nakagami-m fading distribution as a special case, the coverage probability has been studied for both independent and correlated interferers case \cite{6171806, paris2013outage, 6661325}. In particular, Rayleigh fading for interferers channel, and $\eta-\mu$ fading for the user channel are assumed in \cite{6171806}. In \cite{paris2013outage},   $\eta-\mu$ fading has been considered assuming  integer values of fading parameter $\mu$ for interferers channel  and arbitrary fading parameter $\mu$  for the user channel.  Also, $\eta-\mu$ fading with integer values of fading parameter $\mu$ for either the user channel or the interferer channel but not for both  are assumed in \cite{6661325}. 

However, to the best of our knowledge, no prior work in open literature has  analytically compared the coverage probability and rate when interferers are independent with the coverage probability and rate when interferers are correlated. In this work, for Nakagami-m fading we compare the  coverage probability when the interferers are independent and non identically distributed  (i.n.i.d.)  with the coverage probability when the interferers are positively correlated\footnote{If $cov(X_i, X_j)\geq 0$ then $X_i$ and $X_j$  are positively correlated RVs, where $cov(X_i, X_j)$ denotes the covariance between $X_i$ and $X_j$ \cite{403769}.} using majorization theory. It is analytically shown that the coverage  probability in presence of correlated interferers is higher than the coverage probability when the interferers  are i.n.i.d., when the user channel's shape parameter is lesser than or equal to one, and the interferers have Nakagami-m fading with arbitrary parameters (i.e., shape parameter can be less than or greater than one).  We also show that when the user channel's shape parameter is greater than one, one cannot say whether coverage 
probability is higher or lower for the correlated case when compared to the independent case, and in some cases coverage probability is higher while in other cases it is lower\footnote{The expression for outage probability given in \cite{6171806, paris2013outage, 6661325} are in terms of multiple series. However, motivated by the bit error probability expression in the multiple antenna system in the presence of generalized fading model \cite {Aalo, 5073732, 5529757, 1315922, 1291791, 966575}, we have given equivalent expressions for coverage probability in terms of Lauricella hypergeometric function,  since that simplifies the analysis required for comparison between i.n.i.d. and correlated interferers case significantly.}.

We further analytically compare the average rate when the interferers are i.n.i.d.  with the average rate when the interferers are correlated using  stochastic ordering theory. It is shown that the average rate in the presence of positively correlated interferers is higher than  the average rate in the presence of i.n.i.d. interferers. Our results show that correlation among interferers is beneficial for the desired user. We briefly discuss how the desired user can exploit this correlation among the interferers to improve its rate in Section VI.  We have also carried out extensive simulations for both the i.n.i.d. interferers case and the correlated interferers case and some of these results are reported in Section VI.
In all the cases, the simulation results match with our theoretical results.  The work done here can be easily extended to the scenario where the user's  channel experience Nakagami-m fading and interfering channel experience $\eta-\mu$ fading with half integer or integer value of parameter $\mu$.
\section{System Model}
We consider a homogeneous macrocell network with hexagonal structure having inter cell site distance $2R$ as shown in Fig. \ref{fig:hexagonal}. The Signal-to-Interference-Ratio  (SIR) $\eta$ of a user located at $r$ meters from the BS is given by
\begin{equation}
\eta(r)=\frac{Pgr^{-\beta}}{\sum\limits_{i\in\phi}Ph_id_i^{-\beta}}, 
\end{equation}   
where $\phi$ denotes the set of interfering BSs and  $N=|\phi|$ denotes the cardinality of the set $\phi$. The transmit power of a BS is denoted by $P$.  A standard path loss model $r^{-\beta}$ is considered, where $\beta\geq 2$ is the path loss exponent. Note that for the path loss model $r^{-\beta}$ to be valid,  it is assumed that users are at least a minimum distance $d$ meter away from the BS. An interference limited network is assumed, and hence the noise power is neglected. The distance between user to tagged BS (own BS) and the $i$th interfering BS is denoted by $r$ and $d_i$, respectively. The user channel's power and the channel power between $i^{th}$ interfering BS and user are gamma distributed (corresponds to Nakagami-m fading) with $g\sim \mathcal{G}(\alpha_u, \lambda_u)$ and $h_i \sim \mathcal{G}(\alpha_i, \lambda_i)$, respectively.  The pdf  of the gamma RV $g$  is given by
\begin{equation}
f_Y(y)=\frac{y^{\alpha_u-1}e^{-\frac{y}{\lambda_u}}}{(\lambda_u)^{\alpha_u} \Gamma(\alpha_u)}, \text{ } y\geq 0
\end{equation}
where, $\alpha_u\geq 0.5$ is the shape parameter, $\lambda_u>0$ denotes the scale parameter, and $\Gamma(.)$ denotes the gamma function.
The coverage probability of a user located at distance $r$ meters from the BS is defined as
\begin{equation}
P_c(T,r)=P(\eta(r)>T)=P\left(\frac{gr^{-\beta}}{I}>T\right)= P\left(I<\frac{gr^{-\beta}}{T}\right)
\label{coverage}
\end{equation}
where $T$ denotes the target SIR, and $I=\sum\limits_{i\in\phi}h_id_i^{-\beta}$. Since $h_i \sim \mathcal{G}(\alpha_i, \lambda_i)$, hence $I$ is the sum of weighted i.n.i.d. gamma variates with weights $d_i^{-\beta}$. We will use the fact that weighted  gamma variates $h'_i=w_ih_i$  can be written as gamma variates with weighted scale parameter i.e., $h'_i\sim \mathcal{G}(\alpha_i, \lambda_i d_i^{-\beta} )$ \cite{DiSalvo}. Thus, $I=\sum\limits_{i\in\phi}h'_i$ is the sum of $N$ i.n.i.d gamma variates. The pdf of the sum of i.n.i.d. and correlated gamma RVs has been extensively studied in \cite{alouini, Moschopoulos, Aalo, Nadarajah_2008, Karagiannidis, Paris, Kalyani, reig2008performance, 4786505} and references therein. In the context of this paper, we use the confluent Lauricella function representation of the pdf of the sum of gamma variates.

 The sum, $X$ of $N$ i.n.i.d. gamma RVs, $h'_i\sim\mathcal{G}(\alpha_i,\lambda'_i)$ where $\lambda'_i=\lambda_i d_i^{-\beta}$ has a pdf given by \cite{exton1976multiple, Aalo, srivastava1985multiple},
\begin{equation}
f_X(x)=\frac{x^{\sum\limits_{i=1}^{N}\alpha_i-1}}{\prod\limits_{i=1}^{N} (\lambda'_i)^{\alpha_i}\Gamma\left(\sum\limits_{i=1}^{N}\alpha_i\right)}\phi_2^{(N)}\left(\alpha_1,\cdots,\alpha_N;\sum\limits_{i=1}^{N}\alpha_i;\frac{-x}{\lambda'_1 },\cdots,\frac{-x}{\lambda'_N }\right), \text{ } x\geq 0 
\end{equation}
where $\phi_2^{(N)}(.)$ is the confluent Lauricella function \cite{exton1976multiple, exton1978handbook, srivastava1985multiple}.
The cumulative distribution function (cdf)  of $X$ is given by
\begin{equation}
F_X(x)=\frac{x^{\sum\limits_{i=1}^{N}\alpha_i}}{\prod\limits_{i=1}^{N}(\lambda'_i )^{\alpha_i}	\Gamma\left(\sum\limits_{i=1}^{N}\alpha_i+1\right)}\phi_2^{(N)}\left(\alpha_1,  \cdots,\alpha_N;\sum\limits_{i=1}^{N}\alpha_i+1;\frac{-x}{\lambda'_1 },\cdots,\frac{-x}{\lambda'_N}\right).
\end{equation}
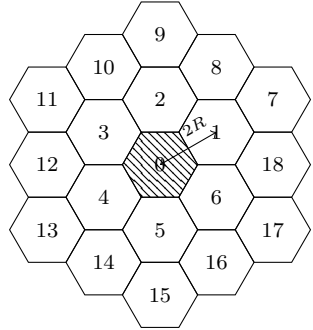
\begin{figure}[ht]
\centering
\begin{tikzpicture}
\node[pattern=north west lines, regular polygon, regular polygon sides=6,minimum width= 1  cm, draw] at (5,5) {};
\node at (5,5){\scriptsize $0$};
\node[regular polygon, regular polygon sides=6,minimum width=1 cm, draw] at (5+1.5*0.5,5-0.866*0.5) {};
\node at (5+1.5*0.5,5-0.866*0.5){\scriptsize $6$};
\node[ regular polygon, regular polygon sides=6,minimum width=1 cm, draw] at (5+1.5*0.5,5+0.866*0.5) {};
\node at (5+1.5*0.5,5+0.866*0.5){\scriptsize $1$};
\node[ regular polygon, regular polygon sides=6,minimum width=1 cm, draw] at (5-1.5*0.5,5+0.866*0.5) {};
\node at (5-1.5*0.5,5+0.866*0.5){\scriptsize $3$};
\node[ regular polygon, regular polygon sides=6,minimum width=1 cm, draw] at (5-1.5*0.5,5-0.866*0.5) {};
\node at (5-1.5*0.5,5-0.866*0.5){\scriptsize $4$};
\node[ regular polygon, regular polygon sides=6,minimum width=1 cm, draw] at (5,5-0.866) {};
\node at (5,5-0.866){\scriptsize $5$};
\node[ regular polygon, regular polygon sides=6,minimum width=1 cm, draw] at (5,5+0.866) {};
\node at (5,5+0.866){\scriptsize $2$};
\node[ regular polygon, regular polygon sides=6,minimum width=1 cm, draw] at (5,5+0.866*2) {};
\node at (5,5+0.866*2){\scriptsize $9$};
\node[ regular polygon, regular polygon sides=6,minimum width=1 cm, draw] at (5,5-0.866*2) {};
\node at (5,5-0.866*2){\scriptsize $15$};
\node[ regular polygon, regular polygon sides=6,minimum width=1 cm, draw] at (5+3*0.5,5+0.866*1) {};
\node at (5+3*0.5,5+0.866*1){\scriptsize $7$};
\node[ regular polygon, regular polygon sides=6,minimum width=1 cm, draw] at (5+3*0.5,5-0.866*1) {};
\node at (5+3*0.5,5-0.866*1){\scriptsize $17$};
\node[ regular polygon, regular polygon sides=6,minimum width=1 cm, draw] at (5-3*0.5,5-0.866*1) {};
\node at (5-3*0.5,5-0.866*1){\scriptsize $13$};
\node[ regular polygon, regular polygon sides=6,minimum width=1 cm, draw] at (5-3*0.5,5+0.866*1) {};
\node at (5-3*0.5,5+0.866*1){\scriptsize $11$};
\node[ regular polygon, regular polygon sides=6,minimum width=1 cm, draw] at (5+1.5,5) {};
\node at (5+1.5,5){\scriptsize $18$};
\node[ regular polygon, regular polygon sides=6,minimum width=1 cm, draw] at (5-1.5,5) {};
\node at (5-1.5,5){\scriptsize $12$};
\node[ regular polygon, regular polygon sides=6,minimum width=1 cm, draw] at (5+1.5*0.5,5-0.866*1.5) {};
\node at(5+1.5*0.5,5-0.866*1.5){\scriptsize $16$};
\node[ regular polygon, regular polygon sides=6,minimum width=1 cm, draw] at (5-1.5*0.5,5+0.866*1.5) {};
\node at(5-1.5*0.5,5+0.866*1.5){\scriptsize $10$};
\node[ regular polygon, regular polygon sides=6,minimum width=1 cm, draw] at (5-1.5*0.5,5-0.866*1.5) {};
\node at(5-1.5*0.5,5-0.866*1.5){\scriptsize $14$};
\node[ regular polygon, regular polygon sides=6,minimum width=1 cm, draw] at (5+1.5*0.5,5+0.866*1.5) {};
\node at(5+1.5*0.5,5+0.866*1.5){\scriptsize $8$};
\draw[<->] (5,5)--(5+1.5*0.5,5+0.866*0.5) node[pos=0.75,sloped,above] {\tiny$2R$};
\end{tikzpicture}
            \caption{ Macrocell network with hexagonal tessellation having inter cell site distance $2R$}
           
             \label{fig:hexagonal}
        \end{figure}
\section{Coverage Probability}
In this section, the coverage probability expression is given in terms of special functions for both the i.n.i.d. interferers case and correlated interferers case. 
\subsection{Coverage Probability in Presence of  i.n.i.d. Nakagami-m Fading}
The coverage probability expression can be written as $P\left(I<\frac{gr^{-\beta}}{T}\right)$. Using the fact that $I$ is the sum of $N$ i.n.i.d. gamma variates, one obtains,
\begin{equation}
P_c(T,r)=E_g\left[\frac{\left(\frac{gr^{-\beta}}{T}\right)^{\sum\limits_{i=1}^{N}\alpha_i}}{	\Gamma\left(\sum\limits_{i=1}^{N}\alpha_i+1\right)\prod\limits_{i=1}^{N}(\lambda'_i )^{\alpha_i}}\phi_2^{(N)}\left(\alpha_1,  \cdots,\alpha_N;\sum\limits_{i=1}^{N}\alpha_i+1;\frac{-gr^{-\beta}}{T\lambda'_1},\cdots,\frac{-gr^{-\beta}}{T\lambda'_N}\right)\right]\label{eq:cov1},
\end{equation}
where $E_g$ denotes expectation with respect to RV $g$ which is gamma distributed. Using transformation of variables with $\frac{g}{\lambda}=t$, and the fact that $g\sim \mathcal{G}(\alpha_u,\lambda_u)$, \eqref{eq:cov1} can be further simplified as 
\begin{equation}
P_c(T,r)=K'\lambda^{\sum\limits_{i=1}^{N}\alpha_i+\alpha_u}\int\limits_{0}^{\infty}\frac{t^{\sum\limits_{i=1}^{N}\alpha_i+\alpha_u-1}e^{-t}}{\Gamma\left(\sum\limits_{i=1}^{N}\alpha_i+\alpha_u\right)}\phi_2^{(N)}\left(\alpha_1, \cdots\alpha_N;\sum\limits_{i=1}^{N}\alpha_i+1;\frac{-t\lambda r^{-\beta}}{T\lambda'_1},\cdots\frac{-t\lambda r^{-\beta}}{T\lambda'_N}\right)\text{d}t\label{eq:cov2}.
\end{equation}
$\text{where, }K'=\frac{1}{\lambda^\alpha_u}\frac{\Gamma\left(\sum\limits_{i=1}^{N}\alpha_i+\alpha_u\right)}{\Gamma\left(\sum\limits_{i=1}^{N}\alpha_i+1\right)}\frac{1}{\Gamma{(\alpha_u)}}\prod\limits_{i=1}^{N}\left(\frac{1}{\lambda'_i r^{\beta}T}\right)^{\alpha_i}$. In order to  simplify \eqref{eq:cov2}, we use the following integral equation \cite[P. 286, Eq 43]{srivastava1985multiple} 
\begin{equation}
F_D^{(N)}[\alpha_u,\beta_1,\cdots,\beta_N;\gamma;x_1,\cdots, x_N]=\frac{1}{\Gamma(\alpha_u)}\int\limits_{0}^{\infty}e^{-t}t^{\alpha_u-1}\phi_2^{N}[\,\beta_1,\cdots,\beta_2,\gamma;x_1t,\cdots,x_Nt]\text{d}t,\label{eq:relationship}
\end{equation}
\begin{equation*}
 \max\{Re(x_1),\cdots,Re(x_N)\}<1, Re(\alpha_u)>0;
\end{equation*}
Here $F_D^{(N)}\left[a,b_1, \cdots, b_N;c;x_1,\cdots, x_N\right]$ is the  Lauricella's function of the fourth kind \cite{mathai1978h}. Using \eqref{eq:relationship}  to  evaluate \eqref{eq:cov2}, one obtains
\begin{equation}
\textstyle
P_c(T,r)=\frac{\Gamma\left(\sum\limits_{i=1}^{N}\alpha_i+\alpha_u\right)}{\Gamma\left(\sum\limits_{i=1}^{N}\alpha_i+1\right)}\frac{1}{\Gamma{(\alpha_u)}}\prod\limits_{i=1}^{N}\left(\frac{\lambda}{\lambda'_i r^{\beta}T}\right)^{\alpha_i} F_D^{(N)}\left[\sum\limits_{i=1}^{N}\alpha_i+\alpha_u,\alpha_1,  \cdots,\alpha_N;\sum\limits_{i=1}^{N}\alpha_i+1;\frac{-\lambda r^{-\beta}}{T\lambda'_1},\cdots,\frac{-\lambda r^{-\beta}}{T\lambda'_N}\right]\label{eq:cov3},
\end{equation}
$F_D^{(N)}( .)$ can be evaluated by using single integral expression \cite{mathai1978h, Aalo} or  multiple integral expression  \cite{exton1976multiple}.
A series expression for $F_D^{(N)}( .)$ involving N-fold infinite sums is given by
\begin{equation} 
F_D^{(N)}[a,b_1,\cdots, b_N;c;x_1,\cdots, x_N]=\sum\limits_{i_1\cdots i_N=0}^{\infty}\frac{(a)_{i_1+\cdots+i_N}(b_1)_{i_1}\cdots(b_N)_{i_N}}{(c)_{i_1+\cdots+i_N}}\frac{x_1^{i_1}}{i_1!}\cdots\frac{x_N^{i_N}}{i_N!},\label{eq:lauricella1}
\end{equation}
\begin{equation*}
\max\{|x_1|,\cdots|x_N|\}<1,
\end{equation*}
where, $(a)_n$ denotes the Pochhammer symbol which is defined as $(a)_n=\frac{\Gamma(a+n)}{\Gamma(a)}$.
 The series expression for Lauricella's function of the fourth kind converges if $\max\{|x_1|,$ $\cdots|x_N|\}<1$. However from \eqref{eq:cov3} it is apparent that convergence condition, i.e., $\underset{i}{\max}{|\frac{-\lambda r^{-\beta}}{T\lambda'_i}|}<1$ is not always satisfied, since $r<d_i$. Hence in order to obtain a series expression for $F_D^{(N)} (.)$ which converges, we use the following property of the Lauricella's function of the fourth kind \cite[p.286]{exton1976multiple}.
\begin{equation}
\textstyle
F_D^{(N)}[a,b_1,\cdots, b_N; c ;x_1,\cdots, x_N]=\left[ \prod_{i=1}^{N}(1-x_i)^{-b_i} \right]F_D^{(N)}\left(c-a,b_1,\cdots, b_N; c ;\frac{x_1}{x_1-1},\cdots, \frac{x_N}{x_N-1}\right) \label{eq:cov4}
\end{equation}
and rewrite   \eqref{eq:cov3}  as
\begin{equation}
\textstyle
P_c(T,r)=\frac{\Gamma\left(\sum\limits_{i=1}^{N}\alpha_i+\alpha_u\right)}{\Gamma\left(\sum\limits_{i=1}^{N}\alpha_i+1\right)\Gamma{(\alpha_u)}} \prod\limits_{i=1}^{N}\left(\frac{\lambda}{\lambda+\lambda'_i r^{\beta}T}\right)^{\alpha_i}F_D^{(N)}\left[1-\alpha_u,\alpha_1,  \cdots, \alpha_N;\sum\limits_{i=1}^{N}\alpha_i+1;\frac{\lambda}{\lambda+ r^{\beta}\lambda'_1T},\cdots, \frac{\lambda}{\lambda+ r^{\beta}\lambda'_NT}\right]\label{eq:cov5}
\end{equation}
\subsection{Coverage Probability in Presence of Correlated Interferers}
In this subsection, we obtain the  coverage probability expression in presence of correlated interferers, when the shape parameter of all the interferers are identical.  The sum, $Z$ of $N$ correlated not necessarily identically distributed gamma RVs $Y_i \sim \mathcal{G}(\alpha_c, \lambda'_i)$ has a cumulative distribution function given by \cite{Paris},\cite{Kalyani},
\begin{equation*}
F_Z(z)=\frac{z^{N\alpha_c}}{	\det(\mathbf{A})^{\alpha_c}\Gamma\left(N\alpha_c+1\right)}\phi_2^{(N)}\left(\alpha_c,\cdots, \alpha_c;N\alpha_c+1;\frac{-z}{\hat{\lambda}_1 },\cdots,\frac{-z}{\hat{\lambda}_N}\right),
\end{equation*}
here, $\mathbf{A=DC}$, where $\mathbf{D}$ is the diagonal matrix with entries $\lambda'_i$ and $\mathbf{C}$ is the symmetric positive definite (s.p.d.) $N\times N$ matrix defined by\\
\begin{equation}
\mathbf{C}=\left[ \begin{array}{cccc}
 1 & \sqrt{\rho_{12}} &...&\sqrt{\rho_{1N}}   \\ \sqrt{\rho_{21}} &1&...&\sqrt{\rho_{2N}} \\  \cdots &\cdots &\ddots & \cdots \\ \sqrt{\rho_{N1}} &\cdots &\cdots &1  \end{array} \right],\label{correlation}
\end{equation}
where $\rho_{ij}$ denotes the correlation coefficient between $Y_i$ and $Y_j$, and is given by,
\begin{equation}
\rho_{ij}=\rho_{ji}=\frac{cov(Y_i,Y_j)}{\sqrt{var(Y_i)var(Y_j)}}, 0\leq\rho_{ij}\leq 1, i,j=1,2, \cdots, N.
\end{equation}
${cov(Y_i,Y_j)}$ and $var (Y_i)$  denote the covariance between  $Y_i$ and $Y_j$ and variance of $Y_i$, respectively. $\det(\mathbf{A})=\prod\limits_{i=1}^{N}\hat{\lambda}_i$  is the determinant of the matrix $\mathbf{A}$, and $\hat{\lambda}_i$s are the eigenvalues of  $\mathbf{A}$. Note that $\hat{\lambda}_i>0 \text{ }\forall i$, since $\mathbf{C}$ is s.p.d. and the diagonal elements of $\mathbf{A}$ are equal to $\lambda'_i$. The functional form of cdf of sum of correlated gamma RVs is similar to the cdf of sum of i.n.i.d. gamma RVs. Hence the coverage probability in the presence of correlated interferers $P_c^c(T,r)$ can be similarly derived  and one obtains
\begin{equation}
\textstyle
P_c^c(T,r)= \frac{\Gamma\left(N\alpha_c+\alpha_u\right)}{\Gamma\left(N\alpha_c+1\right)}\frac{1}{\Gamma{(\alpha_u)}}\prod\limits_{i=1}^{N}\left(\frac{\lambda}{\lambda+\hat{\lambda}_i r^{\beta}T}\right)^{\alpha_c}F_D^{(N)}\left[1-\alpha_u,\alpha_c,  \cdots\alpha_c;N\alpha_c+1;\frac{\lambda}{\lambda+ r^{\beta}\hat{\lambda}_1T},\cdots\frac{\lambda}{\lambda+ r^{\beta}\hat{\lambda_N}T}\right], \label{corr}
\end{equation}
However, note that here the coverage probability is a function of the eigenvalues of $\mathbf{A}$ and the shape parameter of the user and interferer channels while in the i.n.i.d. case it was only a function of the shape parameters and scale parameters.
\section{Comparison of Coverage Probability} 
In this section, we  compare the coverage probability in the i.n.i.d. case and correlated case, and analytically quantify the impact of correlation. Note that the coverage probability expression for the correlated case is derived when the interferers shape parameter are all equal and hence for a fair comparison we consider equal shape parameter for the i.n.i.d. case also, i.e., $\alpha_i=\alpha_c \text{ }\forall i $. 
We first start with the special case when user channel's fading is Rayleigh (i.e, $\alpha_u=1$) and interferers have Nakagami-m fading with arbitrary parameters. When $\alpha_u=1$ and $\alpha_i=\alpha_c \text{ }\forall i $, then the coverage probability in the i.n.i.d. case given in \eqref{eq:cov5}  reduces to
\begin{equation}
P_c(T,r)=\prod\limits_{i=1}^{N}\left(\frac{\lambda}{\lambda+\lambda'_i r^{\beta}T}\right)^{\alpha_c} F_D^{(N)}\left[0,\alpha_c,  \cdots,\alpha_c;\sum\limits_{i=1}^{N}\alpha_c+1;\frac{\lambda}{\lambda+ r^{\beta}\lambda'_1T},\cdots,\frac{\lambda}{\lambda+ r^{\beta}\lambda'_NT}\right].\label{eq:cov6}
\end{equation}
Using the fact   that $(0)_0=1$ and $(0)_k=0 \text{ }\forall \text{ } k\geq 1$,  the coverage probability  is now given by
\begin{equation}
P_c(T,r)=\prod\limits_{i=1}^{N}\left(\frac{1}{1+{\lambda'}_i\frac{r^{\beta} T}{\lambda}}\right)^{\alpha_c}
\end{equation} 
Similarly, the coverage probability in correlated case $P_c^c(T,r)$ is given by
\begin{equation}
P_c^c(T,r)=\prod\limits_{i=1}^{N}\left(\frac{1}{1+\hat{\lambda}_i\frac{r^{\beta} T}{\lambda}}\right)^{\alpha_c}.
\end{equation}
We now state and prove the following theorem for the case where the user channel undergoes Rayleigh fading  and interferers experience  Nakagami-m fading and then generalize it to the  case where user also experiences Nakagami-m fading. 
\newtheorem{lemma}{Lemma} 
\newtheorem{theorem}{Theorem} 
\begin{theorem}
The coverage probability in correlated case is higher than that of the i.n.i.d. case, when user's  channel undergoes Rayleigh fading, i.e.,
\begin{equation}
\prod\limits_{i=1}^{N}\left(\frac{1}{{1+k\hat{\lambda}_i}}\right)^{\alpha_c}\geq \prod\limits_{i=1}^{N}\left(\frac{1}{{1+k\lambda'_i}}\right)^{\alpha_c}
\end{equation}
where $\hat{\lambda}_i$s  are the eigenvalues of  matrix $\mathbf{A}$ and $\lambda'_i$s are the scale parameter for the i.n.i.d. case and $k=\frac{r^{\beta} T}{\lambda}$ is a non negative constant.
\end{theorem}
\begin{proof}
Note that since $\mathbf{A}=\mathbf{D}\mathbf{C}$, the diagonal elements of $\mathbf{A}$ are $\lambda'_i$s. We will briefly state two well known results from   majorization theory\footnote{The notation $\boldsymbol{a}\succ \boldsymbol{b}$  indicate that vector $\boldsymbol{b}$ is majorized by vector $\boldsymbol{a}$. Let $\boldsymbol{a}=[a_1,\cdots a_n]$ and $\boldsymbol{b}=[b_1,\cdots b_n]$ with  $a_1\leq, \cdots,\leq  a_n$ and  $b_1\leq, \cdots,\leq  b_n$ then $\boldsymbol{a}\succ \boldsymbol{b}$ if and only if 
\begin{equation}
\sum\limits_{i=1}^{k}b_i\geq \sum\limits_{i=1}^{k}a_i, \text{ } k=1,\cdots, n-1, \text{ and } \sum\limits_{i=1}^{n}b_i= \sum\limits_{i=1}^{n}a_i. 
\end{equation}} which we will use to prove Theorem $1$.
\begin{theorem}
 If $\boldsymbol{H}$ is an $n \times n$  Hermitian matrix with diagonal elements $b_1,\cdots b_n$  and eigenvalues $a_1,\cdots a_n$ then 
\begin{equation}
\boldsymbol{a} \succ \boldsymbol{b}
\end{equation}
\end{theorem}
\begin{proof}
The details of the proof can be found in  \cite[P. 300, B.1.]{marshall2011inequalities}. 
\end{proof}
In our case, since $\hat{\lambda_i}$s are the eigenvalues and $\lambda'_i$s are the diagonal elements of a symmetric matrix $\mathbf{A}$ hence from Theorem $2$, $\boldsymbol{\hat{\lambda}} \succ \boldsymbol{\lambda'}$ where $\boldsymbol{\hat{\lambda}}=[\hat{\lambda_1}, \cdots, \hat{\lambda_n}]$ and $\boldsymbol{\lambda'}=[{\lambda'_1}, \cdots, {\lambda'_n}]$.

\newtheorem{proposition}{Proposition}
\begin{proposition}
If function $\phi $ is symmetric and convex, then $\phi$ is Schur-convex function. Consequently, $\boldsymbol{x} \succ \boldsymbol {y}$ implies $\phi(\boldsymbol{x})\geq \phi(\boldsymbol{y})$.
\end{proposition}
\begin{proof}
For the details of this proof please refer to \cite[P. 97, C.2.]{marshall2011inequalities}.
\end{proof}

Now if it can be shown that $\prod\limits_{i=1}^{N}\left(\frac{1}{{1+k\lambda'_i}}\right)^{\alpha_c}$ $\left(\text { and }\prod\limits_{i=1}^{N}\left(\frac{1}{{1+k\hat{\lambda}_i}}\right)^{\alpha_c}\right)$ is a Schur-convex function then by a simple application of Proposition $1$ it is evident that $\prod\limits_{i=1}^{N}\left(\frac{1}{{1+k\hat{\lambda}_i}}\right)^{\alpha_c}\geq \prod\limits_{i=1}^{N}\left(\frac{1}{{1+k\lambda'_i}}\right)^{\alpha_c}$. To prove that $\prod\limits_{i=1}^{N}\left(\frac{1}{{1+kx_i}}\right)^{\alpha_c}$ is a Schur convex function we need to show that it is a symmetric and convex function \cite{marshall2011inequalities}. 

It is apparent that the function $\prod\limits_{i=1}^{N}\left(\frac{1}{{1+kx_i}}\right)^{\alpha_c}$ is a symmetric function due to the fact that any two of its arguments can be interchanged without changing the value of the function.
So we now  need to show  that the function  $f(x_1,\cdots,x_n)=\prod\limits_{i=1}^{N}\left(\frac{1}{1+kx_i}\right)^{a_i} $ is a convex function where $x_i\geq 0$, $a_i > 0$.
The function $f(x_1,\cdots,x_n)$ is convex if and only if its Hessian 
$\nabla^2\ f $ is positive semi-definite  \cite{boyd2004convex}. Now, $\nabla^2\ f $ can be computed as
\begin{equation}  
\nabla^2\ f= k^2 \prod\limits_{i=1}^{N}\left(\frac{1}{1+kx_i}\right)^{a_i}\left[  \begin{array}{cccc}
\frac{a_1( a_1+1)}{(1+kx_1)^2} &\frac{a_1a_2}{(1+kx_1)(1+kx_2)}&\cdots  &\frac{a_1a_n}{(1+kx_1)(1+kx_n)}\\ \frac{a_1a_2}{(1+kx_1)(1+kx_2)} &\frac{a_2( a_2+1)}{(1+kx_2)^2}&\cdots  &\frac{a_2a_n}{(1+kx_2)(1+kx_n)}\\ \cdots &\cdots &\ddots  &\cdots\\ \frac{a_1a_n}{(1+kx_1)(1+kx_n)}&\frac{a_2a_n}{(1+kx_2)(1+kx_n)}&\cdots  &\frac{a_n( a_n+1)}{(1+kx_n)^2} \end{array} \right]\label{Hessian}
\end{equation}
We now need  to show that $\nabla^2\ f$ is a positive semi-definite  matrix.   For a real symmetric matrix $\boldsymbol{M}$, if $\boldsymbol{x}^T\boldsymbol{M}\boldsymbol{x}>0$ for every  $N\times 1$ nonzero real vector $\boldsymbol{x}$, then the matrix $\boldsymbol{M}$ is  positive definite (p.d.) matrix \cite[P. 566]{meyer2000matrix}. We now rewrite the Hessian matrix as sum of two matrices and it is then given by
\begin{equation}
\nabla^2\ f= k^2 \prod\limits_{i=1}^{N}\left(\frac{1}{1+kx_i}\right)^{a_i}[\mathbf{P}+\mathbf{Q}].
\end{equation}
Here,
\begin{equation}
\mathbf{P}=\left[  \begin{array}{cccc}
\frac{a^2_1}{(1+kx_1)^2} &\frac{a_1a_2}{(1+kx_1)(1+kx_2)}&\cdots  &\frac{a_1a_n}{(1+kx_1)(1+kx_n)}\\ \frac{a_1a_2}{(1+kx_1)(1+kx_2)} &\frac{a^2_2}{(1+kx_2)^2}&\cdots  &\frac{a_2a_n}{(1+kx_2)(1+kx_n)}\\ \cdots &\cdots &\ddots  &\cdots\\ \frac{a_1a_n}{(1+kx_1)(1+kx_n)}&\frac{a_2a_n}{(1+kx_2)(1+kx_n)}&\cdots  &\frac{a^2_n}{(1+kx_n)^2} \end{array} \right]
\end{equation} 
and 
\begin{equation}
\mathbf{Q}=\left[  \begin{array}{cccc}
\frac{a_1}{(1+kx_1)^2} &0&\cdots  &0\\ 0 &\frac{a_2}{(1+kx_2)^2}&\cdots  &0\\ \cdots &\cdots &\ddots  &\cdots\\ 0&0&\cdots  &\frac{a_n}{(1+kx_n)^2} \end{array} \right] .
\end{equation} 
Here  by definition $k^2 \prod\limits_{i=1}^{N}\left(\frac{1}{1+kx_i}\right)^{a_i}>0$. If now both $\boldsymbol{P}$ and $\boldsymbol{Q}$ are  p.d. then  $\nabla^2\ f$ is p.d. Note that $\boldsymbol{P}$ can be written as $\boldsymbol{U}^T\boldsymbol{U}$ where  $\boldsymbol{U}=\left[\frac{a_1}{(1+kx_1)},\frac{a_2}{(1+kx_2)},\cdots,\frac{a_n}{(1+kx_n)}\right]$ is a $N\times 1 $ vector. Hence, $\boldsymbol{x}^T\boldsymbol{P}\boldsymbol{x}= \boldsymbol{x}^T(\boldsymbol{U}^T\boldsymbol{U})\boldsymbol{x}=||\boldsymbol{U}\boldsymbol{x}||^2 >0$ for every  $N\times 1$ nonzero real vector $\boldsymbol{x}$. Thus, $\boldsymbol{P}$ is a p.d matrix.  Since $\boldsymbol{Q}$ is a diagonal matrix with positive entries, $\boldsymbol{Q}$ is also a p.d matrix. Since sum of two p.d. matrix is p.d. matrix hence $\boldsymbol{P}+\boldsymbol{Q}$ is a p.d. matrix. Thus, $\nabla^2\ f$  is a p.d. matrix and $f(x_1,\cdots,x_n)=\prod\limits_{i=1}^{N}\left(\frac{1}{1+kx_i}\right)^{a_i} $ is a convex function.

Since $\prod\limits_{i=1}^{N}\left(\frac{1}{{1+kx_i}}\right)^{\alpha_c}$ is a convex function and  a symmetric function therefore, it is a Schur-convex function. We have shown that  $\boldsymbol{\hat{\lambda}}\succ \boldsymbol{\lambda'}$  and $\prod\limits_{i=1}^{N}\left(\frac{1}{1+kx_i}\right)^{\alpha_c}$ is a Schur-convex function. Therefore, from Proposition 1,
 $\prod\limits_{i=1}^{N}\left(\frac{1}{{1+k\hat{\lambda}_i}}\right)^{\alpha_c}\geq \prod\limits_{i=1}^{N}\left(\frac{1}{{1+k\lambda'_i}}\right)^{\alpha_c}$.
\end{proof}
Thus, the coverage probability in the presence of correlation among the interferers  is greater than or equal  to the coverage probability in  the i.n.i.d. case,  when user channel undergoes Rayleigh fading and the interferers shape parameter $\alpha_i=\alpha_c \text{ }\forall i$.
Now, we compare the coverage probability for general case, i.e., when   $\alpha_u$ is arbitrary.
\begin{theorem}
The coverage probability in the presence of the correlated interferers is greater than or equal to the coverage probability in presence of i.n.i.d. interferers, when user channel's shape parameter is less than or equal to $1$, i.e., $\alpha_u \leq 1$. When  $\alpha_u>1$, coverage probability in the presence of i.n.i.d. is not always lesser than the coverage probability in the presence of correlated interferers. \end{theorem}
\begin{proof}
Please see Appendix.
\end{proof}
Summarizing, the coverage probability in the presence of correlated interferers is greater than or equal to the coverage probability in presence of independent interferers, when user channel's shape parameter is less than or equal to $1$, i.e., $\alpha_u\leq 1$. When $\alpha_u>1$, one can not say whether coverage probability is better in correlated interferer case or independent interferer case. Note that when $\alpha_u\leq 1$, usually the interferers $\alpha_i$ is also smaller than $1$. However, the proof we have  holds for both $\alpha_i>1$ and $\alpha_i<1$.
\section{Comparison of Rate}
In this section, we compare the average rate when interferers are i.n.i.d with the average rate when the interferers are correlated. We first start with the special case, where $\alpha_u \leq 1$, while the interferer's shape parameter is arbitrary. Then,  the general case is analysed, i.e.,  when user's shape parameter and interferers shape parameter both can be arbitrary, and the scale parameters are also arbitrary.  
\subsection{When user's shape parameter is less  than or equal to $1$.}
The average rate of a user at a distance $r$ is $R(r)=E[\ln(1+\eta(r))]$. Using the fact that for a positive RV $X$, $E[X]=\int\limits_{t>0}P(X>T)\text{d}t$, one obtains
\begin{equation}
R(r)=\int\limits_{t>0}P[\ln(1+\eta(r))>t]\text{d}t.
\end{equation}
\begin{equation}
R(r)\stackrel{(a)}{=}\int\limits_{t>0}P[\eta(r)>e^t-1]\text{d}t.\label{compare_rate1}
\end{equation}
Here $(a)$ follows from the fact that $\ln(1+\eta(r))$ is a monotonic increasing function for $\eta(r)$. Similarly, for correlated case, average rate at distance $r$, $\hat{R}(r)$ is given by
\begin{equation}
\hat{R}(r)=\int\limits_{t>0}P[\hat{\eta}(r)>e^t-1]\text{d}t.\label{compare_rate2}
\end{equation}
Here $\hat{\eta}(r)$ denotes the SIR experienced by the user when interferers are correlated. Now, we compare  $R(r)$ and $\hat{R}(r)$ when $\alpha_u\leq 1$ to see the impact of correlation on the average rate. The integrands of \eqref{compare_rate1} and \eqref{compare_rate2}, i.e.,  $P[\eta(r)>e^t-1]$ and $P[\hat{\eta}(r)>e^t-1]$ are equivalent to the coverage probability expressions for independent interferers case and for correlated interferers case evaluated at $T=e^t-1$, respectively. It has been shown in Theorem $3$ that the coverage probability in the presence of the correlated interferers is greater than or equal to the coverage probability in the presence of independent interferers, when $\alpha_u\leq 1$. In other words, $P[\hat{\eta}(r)>e^t-1]\geq P[\eta(r)>e^t-1]$, $\forall \text{ }t$ when $\alpha_u \leq 1$. Therefore, it is apparent from \eqref{compare_rate1} and \eqref{compare_rate2} that $\hat{R}(r)>R(r)$, when $\alpha_u\leq 1$ since for both integration is over the same interval.

Now, we will compare the average rate when both $\alpha_u$ and $\alpha_i=\alpha_c \text{ } \forall i$ are arbitrary. It is difficult to compare the rate  using the approach  given  above for $\alpha_u\leq 1$ since coverage probability in the presence of the correlated interferers can be greater or lower to the coverage probability in the presence of independent interferers, when $\alpha_u> 1$ (See Appendix for more detail). Hence we compare the average rate using  stochastic ordering theory.
\subsection{When both user's shape parameter and  interferer's shape parameter are arbitrary}
In this subsection, we compare $R=E[\ln(1+\frac{S}{I})]$, and $\hat{R}=E[\ln(1+\frac{S}{\hat{I}})]$ for arbitrary values of shape parameter. Here $\eta(r)=\frac{S}{I}$ and  $\hat{\eta}(r)=\frac{S}{\hat{I}}$, where $S$ is the desired user channel power. Using iterated expectation one can rewrite the rates as
\begin{equation}
R=E_S\left[E_I\left[\ln\left(1+\frac{S}{I}\right)\bigg|S=s\right]\right], \text{ and } \hat{R}=E_S\left[E_{\hat{I}}\left[\ln\left(1+\frac{S}{\hat{I}}\right)\bigg|S=s\right]\right].
\end{equation}
Since the  expectation operator preserves inequalities, therefore if we can show that   $E_{\hat{I}}\bigg[\ln\bigg(1+\frac{S}{\hat{I}}\bigg)$ $\bigg|S=s\bigg]\geq E_I\left[\ln\left(1+\frac{S}{I}\right)\bigg|S=s\right]$, then this implies $\hat{R}\geq R$.

 Here $I$ and $\hat{I}$ are the sum of independent and  correlated interferers, respectively. The sum of interference power in the i.n.i.d. case can be written as  
\begin{equation}
I=\sum\limits_{i=1}^{N}h'_i=\sum\limits_{i=1}^{N}\lambda'_iG_i \text{ with } h'_i\sim \mathcal{G}(\alpha_c,\lambda'_i) \text{ and } G_i\sim \mathcal{G}(\alpha_c,1)
\end{equation}
Similarly, for the correlated case,
\begin{equation}
\hat{I}=\sum\limits_{i=1}^{N}\hat{h}_i=\sum\limits_{i=1}^{N}\hat{\lambda}_iG_i, \label{corre}
\end{equation}
where $\hat{h}_i\sim \mathcal{G}(\alpha_c,\lambda'_i)$.  Recall that these $\hat{h}_i$  are correlated with the correlation structure defined by correlation matrix $\boldsymbol{C}$ given in \eqref{correlation}, and  $\hat{\lambda}_i$s are the eigenvalues of the matrix $\boldsymbol{A}=\boldsymbol{D}\boldsymbol{C}$. In other words one can obtain a correlated sum of gamma variates by multiplying independent and identical distributed (i.i.d.) gamma variates with weight $\hat{\lambda}_i$s.  We now briefly state the theorems in stochastic order theory  that we will use to show that $\hat{R}$ is always greater than equal to $R$.   
\begin{theorem}
Let $X_{1}, X_{2},\cdots, X_{N}$ be exchangeable RVs. Let $\boldsymbol{a}=(a_1,a_2,\cdots,a_N)$ and $\boldsymbol{b}=(b_1,b_2,\cdots,b_N)$ be two vectors of constants. If $\boldsymbol{a}\prec \boldsymbol{b}$, then
\begin{equation}
\sum\limits_{i=1}^{N}a_iX_i\leq_{cx}\sum\limits_{i=1}^{N}b_iX_i.
\end{equation}
\end{theorem}
\begin{proof}
The details of the proof is given in \cite[Theorem 3.A.35]{shaked2007stochastic}.
\end{proof}
Here the notation $X \leq_{cx} Y$ denote that $X$ is smaller than $Y$ in convex order\footnote{If $X$ and $Y$ are two RVs such that $E[\phi(X)]\leq E[\phi(Y)] $ for all convex function $\phi: \mathbb{R}\rightarrow \mathbb{R}$,  provided the expectation exist. Then $X$ is said to be smaller than $Y$ in the convex order.}.  Also, note that a sequence of RVs $X_1,\cdots X_N$ is said to be exchangeable if for all $N$ and $\pi\in S(N)$ it holds that $X_1\cdots X_N \stackrel{\mathcal{D}}{=} X_{\pi(1)}\cdots X_{\pi(N)}$ where $S(N)$ is the group of permutations of $\{1,\cdots N\}$ and $\stackrel{\mathcal{D}}{=}$ denotes equality in distribution \cite{exch}. Furthermore, if $X_i$s are identically distributed,  they are exchangeable \cite[P. 129]{shaked2007stochastic}. Hence $G_i$s are exchangeable since they are identically distributed. It has already been shown that  $\boldsymbol{\hat{\lambda}} \succ\boldsymbol{ \lambda'}$ in Section IV.  Hence by a direct application of  Theorem $4$, one obtains, $I \leq_{cx} \hat{I}$.
\begin{theorem}
If $X\leq_{cx} Y$ and $f(.)$ is convex, then $E[f(X)]\leq E[f(Y)]$.
\end{theorem}
\begin{proof}
The details of the proof is given in \cite[Theorem 7.6.2]{kaas2001modern}.
\end{proof}
$\ln(1+\frac{k}{x})$ is a convex function when $k\geq 0$ and $x \geq 0$ due to the fact that double differentiation of $\ln(1+\frac{k}{x})$ is always non negative, i.e., $\frac{\partial \frac{\partial \ln \left(\frac{k}{x}+1\right)}{\partial x}}{\partial x}=\frac{k (k+2 x)}{x^2 (k+x)^2}\geq 0 $. Note that $S$ and $I$ are non negative RVs, hence by a direct application of Theorem $5$, one obtains 
\begin{equation}
E_{\hat{I}}\left[\ln\left(1+\frac{S}{\hat{I}}\right)\bigg|S=s\right]\geq E_I\left[\ln\left(1+\frac{S}{I}\right)\bigg|S=s\right]
\end{equation}
Since expectation preserve inequalities therefore, $E_S[E_I[\ln(1+\frac{S}{I})]] \leq E_S[E_{\hat{I}}[\ln(1+\frac{S}{\hat{I}})]]$. In other words, positive correlation among the interferers increases the average rate.

Summarizing, the average rate in the presence of the positive correlated interferers is always greater than or equal to the average rate in the presence of independent interferers. Now we briefly discuss the utility of our results in the presence of log normal fading.
\subsection{Log Normal Shadowing}
Although all the analysis so far  (comparison of the coverage probability and average rate) considered only small scale fading and path loss, the analysis can be further extended to take into account shadowing effects. In general, the large scale fading, i.e, log normal shadowing   is modeled by zero-mean log-normal distribution which is given by,
\begin{equation*}
f_X(x)=\frac{1}{x\sqrt{2\pi(\frac{\sigma_{dB}}{8.686})^2}}\exp\left(-\frac{\ln^2(x)}{2(\frac{\sigma_{dB}}{8.686})^2}\right), x>0,
\end{equation*}
where $\sigma_{dB}$ is the shadow standard deviation represented in dB. Typically the value of $\sigma_{dB}$ varies from $3$ dB to $10$ dB \cite{3gpp},\cite{3gpp1}. It is shown in \cite{generalized} that the pdf of the  composite fading channel (fading and shadowing) can be expressed using the generalized-K (Gamma-Gamma) model. Also in \cite{moment}, it has been shown that the generalized-K pdf can be well approximated by Gamma pdf $\mathcal{G}(\alpha_l, \lambda_l)$ using the moment matching method, with $\alpha_l$ and $\lambda_l$ are given by 
\begin{equation}
\alpha_l=\frac{1}{(\frac{1}{\alpha_u}+1)\exp((\frac{\sigma_{dB}}{8.686})^2)-1}=\frac{\alpha_u}{({\alpha_u}+1)\exp((\frac{\sigma_{dB}}{8.686})^2)-\alpha_u} \label{eq:shadow1} 
\end{equation} 
\begin{equation}  
\text{ and } \lambda_l=(1+\alpha_u)\lambda_u\exp\left(\frac{3(\frac{\sigma_{dB}}{8.686})^2}{2}\right)-\alpha_u\lambda_u\exp\left(\frac{(\frac{\sigma_{dB}}{8.686})^2}{2}\right) \label{eq:shadow}
\end{equation} 
Thus, SIR $\eta_l$ of a user can now given by 
\begin{equation}
\eta_l(r)=\frac{Pg_lr^{-\beta}}{\sum\limits_{i\in \phi}Ph^l_{i}d_i^{-\beta}}
\end{equation}
where $g_l\sim \mathcal{G}(\alpha_l, \lambda_l)$ and $h^l_{i}\sim \mathcal{G}(\alpha^l_i, \lambda^l_i)$. Here  $\alpha^l_i=\frac{1}{(\frac{1}{\alpha_i}+1)\exp((\frac{\sigma_{dB}}{8.686})^2)-1} $ and $\lambda^l_i=(1+\alpha_i)\lambda_i\exp\left(\frac{3(\frac{\sigma_{dB}}{8.686})^2}{2}\right)-\alpha_i\lambda_i\exp\left(\frac{(\frac{\sigma_{dB}}{8.686})^2}{2}\right)$. One can now derive the coverage probability expression in the presence of  log-normal shadowing $P_c^l(T,r)$ using the methods  in Section II to obtain,
\begin{equation*}
\textstyle
P_c^l(T,r)=\frac{\Gamma\left(\sum\limits_{i=1}^{N}\alpha_i^l+\alpha_l\right)}{\Gamma\left(\sum\limits_{i=1}^{N}\alpha_i^l+1\right)\Gamma{(\alpha_l)}} \prod\limits_{i=1}^{N}\left(\frac{\lambda_l}{\lambda_l+\lambda_i^l d_i^{-\beta} r^{\beta}T}\right)^{\alpha_i^l} \times
\end{equation*}
\begin{equation}
\textstyle 
F_D^{(N)}\left[1-\alpha_l,\alpha_1^l,  \cdots, \alpha_N^l;\sum\limits_{i=1}^{N}\alpha_i^l+1;\frac{\lambda_l}{\lambda_l+ r^{\beta}\lambda_1^l d_1^{-\beta}T},\cdots, \frac{\lambda_l}{\lambda_l+ r^{\beta}\lambda_N^l d_N^{-\beta}T}\right]\label{eq:cov0}
\end{equation}
Further, the correlation coefficient between two identically distributed generalized-K RVs is derived in \cite[Lemma 1]{5501945}, and it is in terms of correlation coefficient  of the RVs corresponding to the short term fading component ($\rho_{i,j}$) and the correlation coefficient of the RVs corresponding to the shadowing component($\rho^s_{i,j}$). The resultant correlation coefficient ($\rho^l_{i,j}$) is then given by
\begin{equation}
\rho^l_{i,j}=\frac{\frac{\rho_{i,j}}{\left( \exp(\frac{\sigma_{dB}}{8.686})^2)-1\right)}+\rho^s_{i,j} \alpha_c+\rho_{i,j}\rho^s_{i,j}}{\alpha_c+\frac{1}{\left( \exp(\frac{\sigma_{dB}}{8.686})^2)-1\right)}+1}\label{rho}
\end{equation}
Now, similar to the independent case, the coverage probability for correlated interferers case $\tilde{P_c}^l(T,r)$ is given by
given by 
\begin{equation}
\textstyle
\tilde{P_c}^l(T,r)=\frac{\Gamma\left(N\alpha_c^l+\alpha_l\right)}{\Gamma\left(N\alpha_c^l+1\right)\Gamma{(\alpha_l)}} \prod\limits_{i=1}^{N}\left(\frac{\lambda_l}{\lambda_l+\hat{\lambda}_i^l  r^{\beta}T}\right)^{\alpha_i^l} \times
\end{equation}
\begin{equation}
\textstyle
F_D^{(N)}\left[1-\alpha_l,\alpha_c^l,  \cdots, \alpha_c^l;N\alpha_c^l+1;\frac{\lambda_l}{\lambda_l+ r^{\beta}\hat{\lambda}_1^l T},\cdots, \frac{\lambda_l}{\lambda_l+ r^{\beta}\hat{\lambda}_N^l T}\right]\label{eq:cov00}
\end{equation}
where $\hat{\lambda}_i^l$s are the eigenvalues of $\mathbf{A}^l=\mathbf{D}^l\mathbf{C}^l$, where $\mathbf{D}^l$ is the diagonal matrix with entries $\lambda_i^ld_i^{-\beta}$ and $\mathbf{C}^l$ is defined by\\
\begin{equation}
\mathbf{C}^l=\left[ \begin{array}{cccc}
 1 & \sqrt{\rho_{12}^l} &...&\sqrt{\rho_{1N}^l}   \\ \sqrt{\rho_{21}^l} &1&...&\sqrt{\rho_{2N}^l} \\  \cdots &\cdots &\ddots & \cdots \\ \sqrt{\rho_{N1}^l} &\cdots &\cdots &1  \end{array} \right],
\end{equation}
with $\rho_{ij}^l$ is given by \eqref{rho}. Note that both \eqref{eq:cov0} and \eqref{eq:cov00} have a similar functional form and they both are also similar to \eqref{eq:cov5} and \eqref{corr}, respectively.  Hence now  the coverage probability and  average rate for the i.n.i.d. case and correlated case can be compared using the methods outlined in Section IV and Section V. In other words, it can be shown that  the coverage probability in the presence of correlated interferers is greater than or equal to the coverage probability in presence of independent interferers, when user's shape parameter is less than or equal to $1$, i.e., $\alpha_l\leq 1$, in the presence of shadow fading. Also, the average rate in the presence of positive correlated interferers is always greater than or equal to the average rate in the presence of independent interferers, in the presence of shadow fading.
\subsection{Extension of this work for $\eta-\mu$ fading}
Recently, the  $\eta-\mu$ fading distribution with two shape parameters $\eta$ and $\mu$ has been proposed to model a general non-line-of-sight propagation scenario \cite{4231253}.  It includes Nakagami-q (Hoyt), one sided Gaussian, Rayleigh and Nakagami-m as special cases. It has been shown in \cite{5345735} that the sum of correlated $\eta-\mu$ power RVs with half integer or integer value of parameter $\mu$ can be represented by the sum of independent gamma RVs with suitable parameters. Hence,  our analysis on the impact of correlation on the coverage probability and average rate can be extended to the scenarios where the user's  channel experience Nakagami-m fading and interfering channel experience $\eta-\mu$ fading with half integer or integer value of parameter $\mu$. Although, there is a restriction  on the parameter $\mu$, it still entitles us to include one-sided Gaussian, Rayleigh, Nakagami-q (Hoyt) and Nakagami-m (with integer m) fading for the interfering signal in our analysis.

 In the next section, we will show simulation results and discuss how those match with the theoretical results. We also briefly discuss how the analysis carried out in this work can be of utility to the network and the user.
\section{Numerical Analysis and Application}
In this section, we give some simulation results for the coverage probability and rate for both independent and correlated case.  The impact of correlation among interferers on the coverage probability and rate  is  discussed and it is observed that in all simulations the rate is higher for the correlated case when compared to i.n.i.d. case.   
      \begin{figure}[ht]
       \centering
       \includegraphics[scale=0.32]{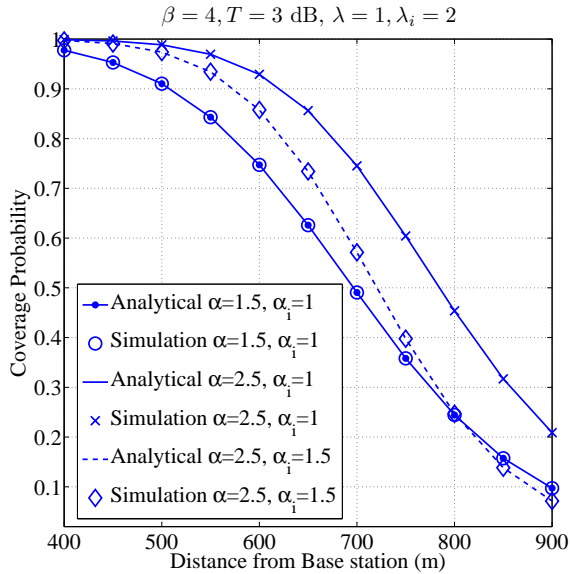}
       \caption{Coverage Probability of a user with respect to distance from the BS in presence of Nakagami-m fading}
       \label{fig:cov}
       \end{figure}
For the simulations, we consider a $19$ cell system with hexagonal structure having inter cell site distance $2R=1732 $ meters as shown in Fig. \ref{fig:hexagonal}. For each user which is connected to the $0$th cell we generate the gamma RV corresponding to its own channel and gamma RVs corresponding to the $18$ interferers and then compute SIR. Then, using the simulated SIR, the coverage probability and average rate can be obtained and they are averaged over 10000 times.

 Fig. \ref{fig:cov} shows the impact of shape parameter on the  coverage probability in the i.n.i.d. case.  We first note that the simulation results exactly match with the analytical results (computed using Eq. \eqref{eq:cov5} ). Secondly, it can be observed that as user channel's shape parameter $(\alpha_u)$ increases while keeping the interferer shape parameters fixed, the coverage probability increases. Whereas, when interferer channel's  shape parameter increases and the user channel's shape parameter is fixed, the coverage probability decreases as expected.

Fig. \ref{fig:correlation} and Fig. \ref{fig:correlation1} depict the impact of correlation among the interferers on the coverage probability for different values of shape parameter. The  correlation among the interferers is defined by the correlation matrix in \eqref{correlation}  with  $\rho_{pq}=\rho^{|p-q|}$ where $p,q=1,\cdots ,N$ \cite{reig2008performance}. From Fig. \ref{fig:correlation}, it can be observed that for  $\alpha_u=0.5$ and $\alpha_u=1$, coverage probability in presence of correlation is higher than that of independent scenario (which match our analytical result).  For example, at $\alpha_u=0.5$, coverage probability increases from $0.12$ in the i.n.i.d case to $0.22$ in the correlated case and at $\alpha_u=1$, coverage probability increases from $0.07$ to $0.12$ when user is at distance $900$m from the BS.   In  Fig. \ref{fig:correlation1}, where $\alpha_u>1$, one cannot say  that coverage probability in presence of correlation is higher or lower than that of independent scenario. However, it can be seen that when  $\alpha_u$  is significantly higher than the interferers  shape parameter (i.e., user channel sees less fading than interferers channel), coverage probability in the presence of independent interferers dominates  over the coverage probability in the presence of correlated interferers. While if $\alpha_u$ is comparable to the interferers channel shape parameter, the coverage probability of independent interferers is higher than the coverage probability of correlated interferers when user is close to the BS. However, the coverage probability of independent interferers is significantly lower than the coverage probability of correlated interferers when the user is far from the BS.
            \begin{figure}[ht]
            \centering
            \includegraphics[scale=0.32]{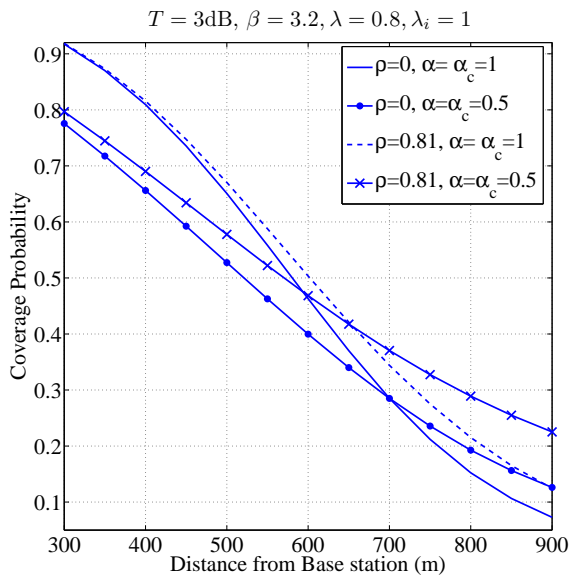}
            \caption{Impact of correlation among the interferers on the coverage probability for different value of shape parameter, when $\alpha_u\leq 1$.}
             \label{fig:correlation}
             \end{figure}
            \begin{figure}[ht]
            \centering
            \includegraphics[scale=0.35]{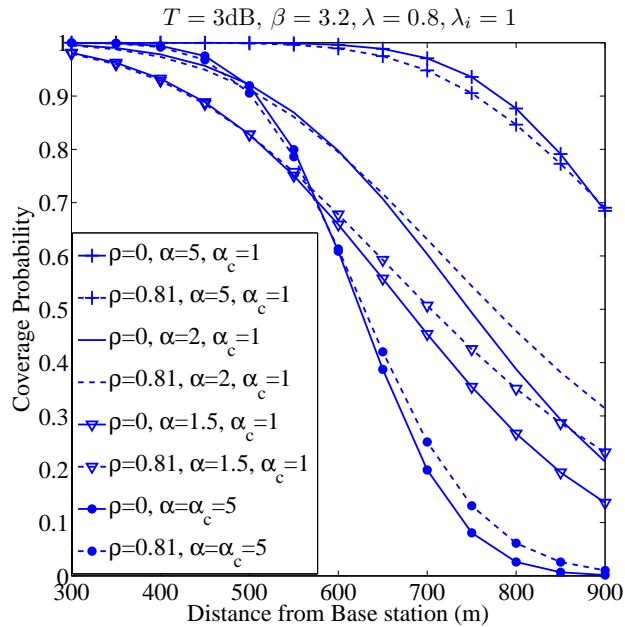}
            \caption{Impact of correlation among the interferers on the coverage probability for different value of shape parameter, when $\alpha_u> 1$.}
             \label{fig:correlation1}
             \end{figure}
\subsection{How the User can Exploit Correlation among Interferers}
We will now briefly discuss how the user in a cellular network can exploit knowledge of positive correlation among its interferers.
We  compare the coverage probability in the presence of correlated interferers for single input single output (SISO) network with the coverage probability in the presence of independent interferers for single input multiple output (SIMO) network to show that the impact of correlation is significant for the cell-edge users (users not near the BS).  For the SIMO network, it is assumed that each user is equipped with $2$ antennas and both antennas at the user are used for reception since downlink is considered. A linear minimum mean-square-error (LMMSE) receiver \cite{tse2005fundamentals} is considered. In order to calculate coverage probability with a LMMSE receiver, it is assumed that the closest interferer can be completely  cancelled  at the SIMO receiver. Fig. \ref{fig:correlation2} plots the SISO coverage probability in the presence of correlated interferers case and the coverage probability in the presence of i.n.i.d. interferers for a SIMO network. It can be seen that for $\rho=0.98$, the SISO coverage probability for the correlated case\footnote{The  correlation among the interferers is defined by the correlation matrix in \eqref{correlation}  with  $\rho_{pq}=\rho^{|p-q|}$ where $p,q=1,\cdots ,N$}  is higher than the SIMO coverage probability for i.n.i.d. case. However, for $\rho=0.81$, SISO coverage probability is close to the SIMO coverage probability at the cell-edge. For example, the coverage probabilities for $\rho=0.98$,  $\rho=0.81$ and the  SIMO case with i.n.i.d. interferers are $0.256$, $0.2$ and $0.18$, respectively, when user is at distance $900$m from the BS.  In other words, correlation among the interferers seems  to be as good as having one additional antenna at the receiver capable of cancelling the dominant interferer.
            \begin{figure}[ht]
            \centering
            \includegraphics[scale=0.32]{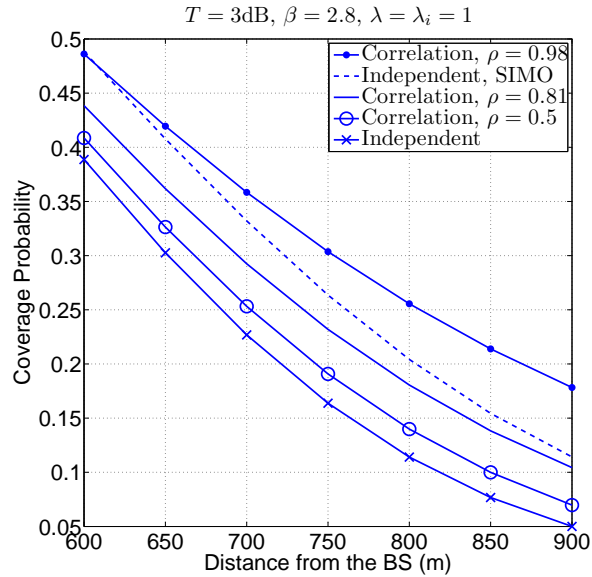}
            \caption{Comparison of coverage probability of correlated interferers with the coverage probability of SIMO network. Here $\alpha_u=\alpha_i=1$.}
             \label{fig:correlation2}
 \end{figure}

Fig. \ref{fig:corr} shows the average rate in the presence of correlation among the interferers and the average rate in the presence of independent interferers for SISO case.   It can be observed that average rate is higher in the presence of correlated interferers. It can be also seen that for $\rho=0.98$, average rate for the correlated case is higher than that of the $1 \times 2$ SIMO network. For example, average rates for $\rho=0.98$, the SIMO case and independent case are $1.41$ nats/Hz, $1.28$ nats/Hz and $1.09$ nats/Hz, respectively, when user is at distance $600$m from the BS.

 Fig. \ref{fig:corr1} presents the comparison  of average rate in the presence of log normal shadowing.  It can be seen that for both $\rho=0.98$ and $\rho=0.81$, the average rate of the correlated case is higher than that of the  $1 \times 2$ SIMO network with $\rho=0$. In other words, in presence of shadowing the impact  of correlation is even more significant. For example, average rates for SISO case with $\rho=0.98$, the SIMO case with $\rho=0$, and the SISO case with $\rho=0$ are $1.25$ nats/Hz, $0.86$ nats/Hz and $0.71$ nats/Hz, respectively when the user is at distance $700$m from the BS. Obviously, if one had correlated interferers in the SIMO system that would again lead to improved coverage probability and average rate and may be compared to a SIMO system with higher number of antennas. In all three cases, it is apparent that if the correlation among the interferers is exploited, it leads to  performance results for a SISO system which are comparable to the   performance of a $1\times 2$ SIMO system with independent interferers.
        \begin{figure}[ht]
            \centering
            \includegraphics[scale=0.32]{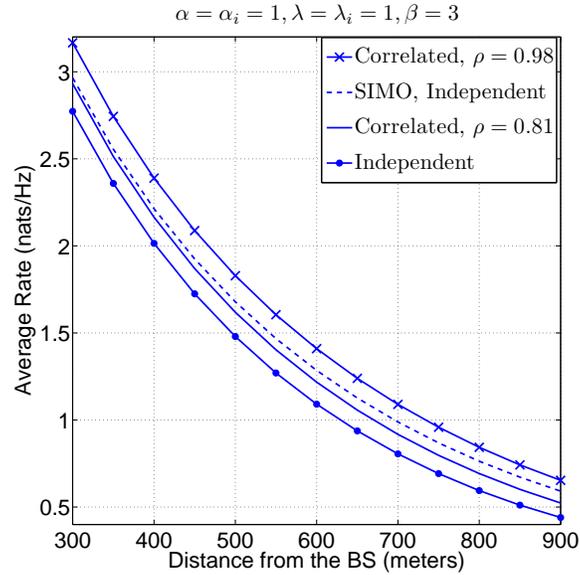}
            \caption{Comparison of the average rate in presence of independent interferers with the average rate in presence of correlated interferers. }
             \label{fig:corr}
             \end{figure}
        \begin{figure}[ht]
            \centering
            \includegraphics[scale=0.32]{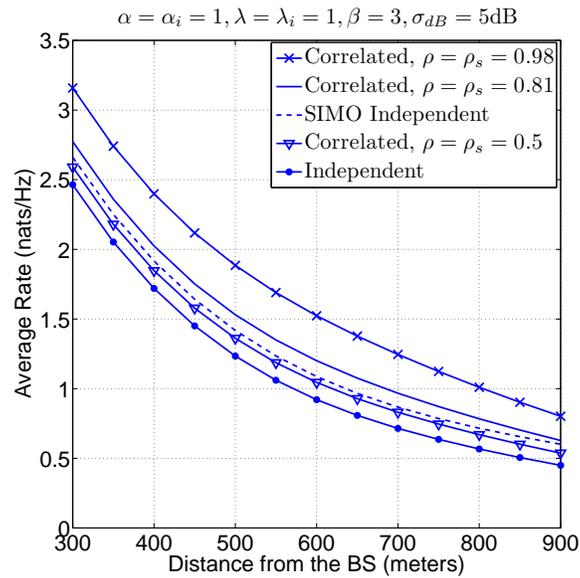}
            \caption{Comparison of the average rate in the  presence of log normal shadowing. }
             \label{fig:corr1}
             \end{figure}

We would also likely to briefly point out that the impact of correlation among the interferers is like that of introducing interference alignment in a system. Interference alignment actually aligns interference using appropriate precoding so as reduce the number of interferers one needs to cancel. Here the physical nature of the wireless channel and the presence of co-located interferers also ``aligns'' the interferer partially. This is the reason one can get a gain equivalent to $1\times 2$ system in a $1\times 1$ system with correlated interferers provided the user knows about the correlation. 

Summarizing, our work is able to analytically shows the impact of correlated interferers on coverage probability and rate. This can be used by the network and user to decide whether one wants to use the antennas at the receiver for diversity gain or interference cancellation depending on the information available about interferers correlation. Note that interferers from adjacent sector of a BS will definitely be correlated \cite{correlation, 944855, 991146, 5590312, 3gpp}. We believe that this correlation should be exploited, since the analysis shows that knowledge of correlation will lead to higher coverage probability and rate.  

\section{Conclusions}
In this work, the coverage probability expressions and rate expressions have been compared analytically for following two cases: $(a)$ Interferers and user channel having arbitrary Nakagami-m fading parameters. $(b)$ Interferers being correlated where the correlation is specified by a correlation matrix. We have  shown that the coverage probability in correlated interferer case is higher than that of the independent case, when the user channel's shape parameter is lesser than or equal to one, and the interferers have Nakagami-m fading with arbitrary parameters. Further, it has been shown that positive correlation among the interferers always increases the average rate.  We have also taken into account the shadow fading component in our analysis. The impact of correlation seems even more pronounced in the presence of  shadow fading. Our results indicate that if the user is aware of the interferers correlation matrix then it can exploit it since the correlated interferers behave like partially aligned interferers. This means that if the user is aware of the correlation then one can obtain a rate equivalent to a $1\times 2$ system in a $1\times 1$ system depending on the correlation matrix structure. Extensive simulations  were performed and these match with the theoretical results.
\appendix
\label{app}
\section*{Proof of Theorem $3$}
The coverage probability expressions for the scenario when interferers are i.n.i.d. and the scenario when interferers are correlated  are given in \eqref{eq:cov5} and \eqref{corr}, respectively and rewriting them for the case when $\alpha_i=\alpha_c \text{ }\forall i$, one obtains
\begin{equation}
P_c(T,r)=K' F_D^{(N)}\left[1-\alpha_u,\alpha_c,  \cdots ,\alpha_c;N\alpha_c+1;\frac{\lambda}{\lambda+ r^{\beta}\lambda'_1T},\cdots,\frac{\lambda}{\lambda+ r^{\beta}\lambda'_NT}\right]\label{compare}
\end{equation}
\begin{equation}
P_c^c(T,r)=\hat{K} F_D^{(N)}\left[1-\alpha_u,\alpha_c,  \cdots,\alpha_c;N\alpha_c+1;\frac{\lambda}{\lambda+ r^{\beta}\hat{\lambda}_1T},\cdots,\frac{\lambda}{\lambda+ r^{\beta}\hat{\lambda_N}T}\right] \label{compare1},
\end{equation}
where $K'=\frac{\Gamma\left(N\alpha_c+\alpha_u\right)}{\Gamma\left(N\alpha_c+1\right)}\frac{1}{\Gamma{(\alpha_u)}}\prod\limits_{i=1}^{N}\left(\frac{1}{1+\lambda'_i \frac{r^{\beta}T}{\lambda}}\right)^{\alpha_c}$ and $\hat{K}=\frac{\Gamma\left(N\alpha_c+\alpha_u\right)}{\Gamma\left(N\alpha_c+1\right)}\frac{1}{\Gamma{(\alpha_u)}}\prod\limits_{i=1}^{N}\left(\frac{1}{1+\hat{\lambda}_i \frac{r^{\beta}T}{\lambda}}\right)^{\alpha_c}$.  From Theorem $1$ it is clear that $\hat{K}> K'$ . Now, we need to compare the Lauricella's function of the fourth kind of \eqref{compare} and \eqref{compare1}. Here, for comparison we use the series expression for $F_D(.)$.

We expand the  series expression for the Lauricella's function of the fourth kind in the  following form:
\begin{align}
F_D^{(N)}[a,b,\cdots, b;c;x_1,\cdots, x_N]=&1+K_{1,1}\sum\limits_{i=1}^{N}x_i+K_{2,1}\sum\limits_{i=1}^{N}x_i^2+K_{2,2}\sum\limits_{1\leq i<j\leq N}x_ix_j+K_{3,1}\sum\limits_{i=1}^{N}x_i^3\nonumber\\
&+K_{3,2}\sum\limits_{i,j=1, s.t.i\neq j}^{N}x_i^2x_j+K_{3,3}\sum\limits_{1\leq i<j<k\leq N}x_ix_jx_k+\cdots\label{lauricella}
\end{align}
 where
$K_{1,1}=\frac{(a)_1(b)_1}{(c)_1 1!}$, $K_{2,1}=\frac{(a)_2(b)_2}{(c)_2 2!}$, $K_{2,2}=\frac{(a)_2(b)_1(b)_1}{(c)_21!1!}$, $K_{3,1}=\frac{(a)_3(b)_3}{(c)_3 3!}$, $K_{3,2}=\frac{(a)_3(b)_2(b)_1}{(c)_3 2!1!}$, $K_{3,3}=\frac{(a)_3(b)_1(b)_1(b)_1}{(c)_31!1!1!}$ and so on.

Hence the coverage probability for independent case given in \eqref{compare} can be written as
\begin{align}
P_c(T,r)=&K'\bigg[ 1+K_{1,1}\sum\limits_{i=1}^{N}\left(\frac{1}{1+\lambda'_i\frac{r^{\beta} T}{\lambda}}\right)+K_{2,1}\sum\limits_{i=1}^{N}\left(\frac{1}{1+\lambda'_i\frac{r^{\beta} T}{\lambda}}\right)^2+ K_{3,1}\sum\limits_{i=1}^{N}\left(\frac{1}{1+\lambda'_i\frac{r^{\beta} T}{\lambda}}\right)^3+\nonumber\\
&K_{2,2}\sum\limits_{1\leq i<j\leq N}\left(\frac{1}{1+\lambda'_i\frac{r^{\beta} T}{\lambda}}\right)\left(\frac{1}{1+\lambda'_j\frac{r^{\beta} T}{\lambda}}\right)+
K_{3,2}\sum\limits_{i,j=1, s.t. i\neq j}^{N}\left(\frac{1}{1+\lambda'_i\frac{r^{\beta} T}{\lambda}}\right)^2\left(\frac{1}{1+\lambda'_j\frac{r^{\beta} T}{\lambda}}\right)+ \nonumber\\
 &K_{3,3}\sum\limits_{1\leq i<j<k\leq N}\left(\frac{1}{1+\lambda'_i\frac{r^{\beta} T}{\lambda}}\right)\left(\frac{1}{1+\lambda'_j\frac{r^{\beta} T}{\lambda}}\right)\left(\frac{1}{1+\lambda'_k\frac{r^{\beta} T}{\lambda}}\right)+\cdots\bigg]\label{sum1}
\end{align} 
Similarly,  for the correlated  case  the coverage probability given in \eqref{compare1} can be written as
\begin{align}
P_c^c(T,r)=&\hat{K}\bigg[ 1+K_{1,1}\sum\limits_{i=1}^{N}\left(\frac{1}{1+\hat{\lambda}_i\frac{r^{\beta} T}{\lambda}}\right)+K_{2,1}\sum\limits_{i=1}^{N}\left(\frac{1}{1+\hat{\lambda}_i\frac{r^{\beta} T}{\lambda}}\right)^2+ K_{3,1}\sum\limits_{i=1}^{N}\left(\frac{1}{1+\hat{\lambda}_i\frac{r^{\beta} T}{\lambda}}\right)^3+\nonumber\\
&K_{2,2}\sum\limits_{1\leq i<j\leq N}\left(\frac{1}{1+\hat{\lambda}_i\frac{r^{\beta} T}{\lambda}}\right)\left(\frac{1}{1+\hat{\lambda}_j\frac{r^{\beta} T}{\lambda}}\right)+
K_{3,2}\sum\limits_{i,j=1, s.t. i\neq j}^{N}\left(\frac{1}{1+\hat{\lambda}_i\frac{r^{\beta} T}{\lambda}}\right)^2\left(\frac{1}{1+\hat{\lambda}_j\frac{r^{\beta} T}{\lambda}}\right)+\nonumber\\
 &K_{3,3}\sum\limits_{1\leq i<j<k\leq N}\left(\frac{1}{1+\hat{\lambda}_i\frac{r^{\beta} T}{\lambda}}\right)\left(\frac{1}{1+\hat{\lambda}_j\frac{r^{\beta} T}{\lambda}}\right)\left(\frac{1}{1+\hat{\lambda}_k\frac{r^{\beta} T}{\lambda}}\right)+\cdots\bigg]\label{sum2}
\end{align} 
Here 
$K_{1,1}=\frac{(1-\alpha_u)_1(\alpha_c)_1}{(N \alpha_c+1)_1 1!}$, $K_{2,1}=\frac{(1-\alpha_u)_2(\alpha_c)_2}{(N \alpha_c+1)_2 2!}$, $K_{2,2}=\frac{(1-\alpha_u)_2(\alpha_c)_1(\alpha_c)_1}{(N \alpha_c+1)_2 1!1!}$, $K_{3,1}=\frac{(1-\alpha_u)_3(\alpha_c)_3}{(N \alpha_c+1)_3 3!}$, $K_{3,2}=\frac{(1-\alpha_u)_3(\alpha_c)_2(\alpha_c)_1}{(N \alpha_c+1)_3 2!1!}$, $K_{3,3}=\frac{(1-\alpha_u)_3(\alpha_c)_1(\alpha_c)_1(\alpha_c)_1}{(N \alpha_c+1)_31!1!1!}$ and so on. Note that here $K_{i,j}$ are the same for both $P_c(T,r)$ and $P_c^c(T,r)$. Now, we want to show that each summation term in the series expression is a Schur-convex function.

Each summation term in the series expression is symmetrical due to the fact that any two of its argument can be interchanged without changing the value of the function. We have already shown that $\prod\limits_{i=1}^{N}\left(\frac{1}{1+kx_i}\right)^{a_i} $ is a convex function $\forall x_i\geq 0$ and $\forall a_i>0$. Now, the terms in the summation terms in \eqref{sum1} and \eqref{sum2} are of the form $\prod\limits_{i=1}^{M}\left(\frac{1}{1+kx_i}\right)^{a_i} $ where $M\leq N$. To show that these functions are convex function we need to show that the corresponding Hessian are p.d. The corresponding Hessians are nothing but principal sub-matrices of the matrix in \eqref{Hessian}. Hence using the fact that every principal sub-matrix of a s.p.d. matrix is a s.p.d. matrix \cite{meyer2000matrix}, one can show that each  term of each summation term is a convex function. Using the fact that convexity is preserved under summation one can show that each summation term is a convex function. Thus, each summation term 
in 
series expression is a Schur-convex function.

Now we consider following two cases.

Case I when $\alpha_u< 1$:
Since $\alpha_u< 1$, so $1-\alpha_u > 0$ and hence all the constant $K_{i, j}> 0 \text{ } \forall\text{ } i, j$.  Each summation term in series expression of coverage probability for correlated case is greater  than or equal to the corresponding summation term in the series expression of coverage probability for independent case. Thus, if user  channel's shape parameter $\alpha_u< 1$ then coverage probability of  correlated case is  greater than or equal to  the coverage probability for independent case.

Case II when $\alpha_u> 1$:
Since $\alpha_u > 1$, then $1-\alpha_u < 0$ and hence $K_{i,j}<0 \text{ }\forall i \in 2|\mathbb{Z}|+1 \text { and } \forall j $  where set $\mathbb{Z}$ denote the integer number, due to  the fact that $(a)_N<0 \text{ if } a<0 \text{ and } N \in 2|\mathbb{Z}|+1 $.  Whereas, $K_{i,j}>0 \text{ }\forall i \in 2|\mathbb{Z}| \text { and } \forall j $ due of the fact that $(a)_N>0 \text{ if } a<0 \text{ and } N\in 2|\mathbb{Z}|$. Thus, if   $\alpha_u>1$, we cannot state whether the coverage probability of one case is greater than or lower than the other case.

\bibliographystyle{IEEEtran}
\bibliography{bibfile}

\end{document}